\documentclass[letterpaper, 10 pt]{article}  

\usepackage[latin1]{inputenc}

\usepackage{psfrag}
\usepackage{tikz}

\usepackage{enumerate}
\usepackage{xstring}
\usepackage{textcomp}
\usepackage{graphics}
\usepackage{graphicx}
\usepackage{caption}
\usepackage{subcaption}
\usepackage{geometry}
\usepackage{etoolbox}
\usepackage{refcount}
\usepackage{amsmath}
\usepackage[toc,page]{appendix}

\usepackage{amsthm}
\newtheorem{ex}{Example}

\newtheorem{thm}{Theorem}

\newtheorem{definition}{Definition}
\newtheorem{prop}{Proposition}

\newtheorem{lem}{Lemma}

\newtheorem{conj}{Conjecture}

\newtheorem{corollary}{Corollary}

\usepackage{cite}

\usepackage{color}
\usepackage{amsfonts}
\usepackage{amssymb}
\usepackage{graphicx}
\usepackage{algorithmicx}
\usepackage{algpseudocode}
\usepackage{color}
\usepackage{amssymb}

\usepackage{xcolor}

\DeclareMathOperator{\sia}{sia}
\DeclareMathOperator{\sar}{sar}
\DeclareMathOperator{\scr}{scr}
\DeclareMathOperator{\pc}{pc}

\title{\LARGE \bf
Sets of Stochastic Matrices with Converging Products: Bounds and Complexity
}
\author{Pierre-Yves Chevalier, Vladimir V. Gusev, Julien M. Hendrickx,\\ and Rapha\"{e}l M. Jungers
}

\begin{document}

\maketitle 
\thispagestyle{empty}

\begin{abstract}
An SIA matrix is a stochastic matrix whose sequence of powers converges to a rank-one matrix. This convergence is desirable in various applications making use of stochastic matrices, such as consensus, distributed optimization and Markov chains.
We study the shortest SIA products of sets of matrices. We observe that the shortest SIA product of a set of matrices is usually very short and we provide a first upper bound on the length of the shortest SIA product (if one exists) of any set of stochastic matrices. 
We also provide an algorithm that decides the existence of an SIA product. 

When particularized to automata, the problem becomes that of finding periodic synchronizing words, and we develop the consequences of our results in relation with the celebrated \v{C}ern\'{y} conjecture in automata theory. 

We also investigate links with the related notions of positive-column, Sarymsakov, and scrambling matrices. 
\end{abstract}

\setcounter{page}{1}

\section{Introduction}
A stochastic matrix $P$ is an \emph{SIA matrix}  (Stochastic\footnote{A matrix is \emph{stochastic} if it is nonnegative and the sum of the elements on each row is 1.} Indecomposable Aperiodic) if the limit
\[ 
\lim_{t \to +\infty} P^t
\] 
exists and all of its rows are equal, i.e., the limit is a rank-one matrix. These matrices play a fundamental role in the theory of Markov chains as they correspond to chains converging to unique limiting distribution~\cite{Paz1963}. They also arise  
in 
discrete-time consensus systems, which are the systems representing groups of agents trying to agree on some value by iterative averaging~\cite{bl-ol}. Such systems can be modelled by the update equation
\begin{equation}x(t+1) = P(t) x(t) \label{eq:cons},\end{equation} where $x(t)$ is the vector of the values of the agents at time $t$ and $P(t)$ is a stochastic transition matrix representing how agents compute their new values.


Markov chains, and consensus systems, are often time-inhomogeneous, i.e., the transition matrix $P(t)$ can be different at different steps $t$, where $P(t)$ typically belongs to a finite set of stochastic matrices $\mathcal{S} = \{A_1, A_2, \ldots, A_k\}$. Such chains arise in diverse settings including flocking~\cite{Jad_tac}, 
simulated annealing in optimization~\cite{annealing87}, Monte Carlo methods~\cite{maire2014comparison} and queueing theory~\cite{Alfa2008}. 

The analysis of convergence is much more difficult in the inhomogenous case than in the time-invariant case. One of the  tools to analyze Markov chains with multiple modes is the classical result of J.~Wolfowitz stating that any product of transition matrices from $\mathcal{S}$ converges to a rank-one matrix if and only if every product of matrices in $\mathcal{S}$ is SIA~\cite{Wolf63}. Sets satisfying this condition are called quasidefinite~\cite{Paz65}, consensus \cite{bl-ol}, or weakly-ergodic matrix sets~\cite{COPPERSMITH}. Informally, the chain represented by such $\mathcal{S}$ forgets its distant past. This theorem has found many applications~\cite{Jad_tac,bl-ol} and guarantees convergence of the modelled stochastic system in \emph{every} possible scenario or  switching between matrices in $\mathcal{S}$.
In the present paper we study sets of matrices $\mathcal{S}$ that have \emph{at least one} converging sequence. While the former condition is likely to appear in a general or adversarial setting, the latter is more relevant in a controlled environment. 

We will say that a set of stochastic matrices $\mathcal{S}$ is \emph{SIA} if there is a product of matrices from $\mathcal{S}$ with repetitions allowed that is an SIA matrix. It is not hard to show (see Section~\ref{sec:classes}) that $\mathcal{S}$ has an infinite product converging to a rank-one matrix if and only if $\mathcal{S}$ is an SIA set. Therefore, SIA sets correspond to stochastic systems that can be brought to convergence by controlling the switching signal. Another motivation to study SIA sets comes from the study of inhomogeneous Markov chains or consensus systems 
\emph{with random switching}~\cite{costa2005discrete}.
If at every step the matrix is chosen independently and with non-zero probability from $\mathcal{S}$, then the system converges with probability one if\footnote{Due to the fact that any finite product of matrices from $\mathcal{S}$ appears infinitely many times with probability one} and only if $\mathcal{S}$ is an SIA set.

For an SIA set $\mathcal{S}$, we denote by $\sia(\mathcal{S})$  the length of the shortest SIA product of matrices from $\mathcal{S}$, and we call this quantity the \emph{SIA-index} of $\mathcal{P}$. The SIA-index is a natural parameter of SIA sets. 
It is similar and related to the classical and well studied notion of \emph{exponent} -- the length of the shortest product that is entrywise positive, if one exists (see~\cite[Section 3.5]{BrualdiRyser1991CombinatorialMatrixTheory} for a survey of the single matrix case, and~\cite{GGJ16,PV} for more recent work on matrix sets). Similar quantities have been defined for different matrix classes. For instance, the scrambling-index is defined as the length of the shortest \emph{scrambling}\footnote{\label{note1}These classes will be defined in Section~\ref{sec:classes}.} 
product of a matrix set, if one exists. In this article, we will define indices for all the considered classes: the classes of Sarymsakov, SIA and positive-column matrices\footnotemark[3].

In the context of a system whose switching sequence can be controlled, an SIA product corresponds to a switching sequence that can be repeated to make the system converge to a rank-one matrix. Thus, shorter SIA products correspond to simpler controllers. Furthermore, the length of the SIA product has an influence on the converging rate. Indeed, if $P = A_\ell \dots A_2 A_1$ is an SIA product, the sequence $\dots A_\ell \dots A_2 A_1 A_\ell \dots A_2 A_1$ converges to a rank-one matrix at an average rate of $\lambda_2^{1/\ell}$, where $\lambda_2$ is the second largest eigenvalue of $P$. 

Finally, the SIA-index brings new insights to synchronizing automata and the \emph{\v{C}ern\'{y} conjecture}~\cite{Cerny64} stating that for any set of stochastic\footnote{Typically, the conjecture is posed for deterministic finite state automata, i.e., matrices having exactly one 1 on each row, but the presented version is equivalent \cite{CDC}.} $n \times n$ matrices $\mathcal{S}$ either there is a product of length at most $(n-1)^2$ of matrices from $\mathcal{S}$ having a positive column or there is no such product at all. The conjecture has been open for half a century and the best bound obtained so far is cubic in $n$~\cite{volkov_survey}.
As we will soon see, a good upper bound on the SIA-index of $n \times n$ matrices would improve the state of the art on the \v{C}ern\'{y} conjecture. In particular, 
any subquadratic bound would bring a breakthrough.



\textbf{Our contribution.} 
We introduce and study the SIA-index of sets of stochastic matrices. We show that the largest value of the SIA-index among all SIA sets of $n \times n$ stochastic matrices, $\sia(n)$, is $O(n^3)$. 
Moreover, we show that $\sia(n)$ grows at least as $n$. 
We conjecture that the actual growth rate is closer to the provided lower bound and support this statement by performing an exhaustive search for small values of $n$ on a computer grid.


We show that a SIA set has a scrambling, a Sarymsakov and a positive-column product, and, conversely, a set that has a scrambling, a Sarymsakov or a positive-column product is a SIA set. 
As a consequence, the same polynomial-time procedure can be used to decide whether a given set of stochastic matrices has any of the aforementioned products. 
We show the SIA-index is NP-hard to compute, even if all matrices in $\mathcal{S}$ have a positive diagonal, and to approximate within a factor of $O(\log(n))$. Matrices with positive diagonals appear in  consensus applications and many problems are computationally easier for these matrices~\cite{CDC}.

Finally, we show that the $(n-1)$st power of any SIA automata matrix has a positive column and that the $(n^2 - 3n + 3)$rd power of any SIA matrix has a positive column. 


We mention that the majority of the presented results only depend on the pattern of nonzero elements in the matrices. Therefore, they also can be applied for \emph{row-allowable matrices}, i.e., nonnegative matrices that have at least one positive element on each row.

\textbf{Paper Organization.} In Section~\ref{sec:classes} we give an overview of commonly used matrix classes: Positive-column, scrambling, Sarymsakov and SIA matrices (Subsection \ref{sec:def}) and show the relations between them (Subsection \ref{sec:rel}). We study the associated indices and bounds on these indices (Subsection \ref{sec:indices}) and derive an upper bound on the power at which an SIA matrix has a positive column (Subsection~\ref{sec:from_to}). Section~\ref{sec:sia-index} is devoted to the properties of the SIA-index. In Subsection~\ref{sec:bound}, lower and upper bounds on the largest SIA-index are provided, and we present an experiment that computes the largest SIA-index for small values of $n$.  In Subsection~\ref{sec:complexity} we discuss the procedure to decide whether a given set of matrices is SIA and the hardness of computing and approximating the SIA-index.


\section{Common Matrix Classes and their Indices}
\label{sec:classes}

In this section, we define different matrix classes that have been introduced in the context of Markov chains and probabilistic automata \cite{Seneta81, Paz71, Wolf63}: positive-column matrices, scrambling matrices, Sarymsakov matrices and SIA matrices. We show that these classes are included in one another and 
that sufficiently large powers of matrices from the larger classes belong to the smaller classes (for example, any sufficiently large power of a Sarymsakov matrix is scrambling).

We also define the associated 
\emph{indices} 
-- the lengths of the shortest products belonging to the corresponding classes.
We relate these indices to each other and present upper bounds on them.

\subsection{Definition of the Matrix Classes}
\label{sec:def}

We start by defining the rather intuitive notion of a positive-column matrix. These matrices are called \emph{Markov matrices} in certain sources \cite[Definition 4.7]{Seneta81}, but we avoid this terminology as the term \emph{Markov matrix} is sometimes used to denote a stochastic matrix. Positive-column matrices have a direct interpretation in many applications. For instance,  in the consensus system \eqref{eq:cons}, an element of the transition matrix $P(t)$ represents the influence of the agent corresponding to the column on that of the row, and a positive-column transition matrix thus corresponds to a situation in which an agent influences all the others. Additionally, positive-column matrices are related to convergence, as the sequence of powers of a positive-column stochastic matrix converges to a rank-one matrix.  
\begin{definition}[Positive-column matrix] A \emph{positive-column matrix} is a  nonnegative matrix that has a positive column. 
We will denote the set of positive-column matrices by $S_\text{PC}$.
\end{definition}

We now define scrambling matrices, a class that has been extensively investigated due to their applications to Markov chains \cite{Seneta81, Paz71}.

\begin{definition}[Scrambling matrix \cite{Paz71}] A nonnegative matrix is called \emph{scrambling} if for any pair of rows $(i,j)$, there is a column $k$ such that $a_{ik} > 0$ and $a_{jk} > 0$. 
We will denote the set of scrambling matrices by $S_\text{SCR}$.
\end{definition}

Scrambling matrices have two interesting properties:
\begin{itemize}
\item the class is closed under multiplication: the product of two scrambling matrices is scrambling;
\item for any $n \in \mathbb{N}$ there is a length $\ell$ such that any product of $\ell$ or more $n \times n$ scrambling matrices has a positive column.
\end{itemize}


The Sarymsakov class of matrices is a larger class that has the same desirable properties~\cite{Sarymsakov61}. 

\begin{definition}[Sarymsakov matrix \cite{Sarymsakov61}] 
We define the \emph{consequent function} of a given nonnegative $n \times n$ matrix $P$. For any $S \subseteq \{1,\dots,n\}$, the consequent function $F$ is equal to
$$F_P(S) = \{j : \exists i \in S \text{ s.t. } a_{ij} > 0\}.$$
 A nonnegative matrix  $P$ is called a \emph{Sarymsakov} matrix if for any two disjoint nonempty subsets $S$ and $S'$,
\begin{equation} F_P(S) \cap F_P(S') \neq \emptyset \label{sarym1}\end{equation}
or \begin{equation} |F_P(S) \cup F_P(S')| > |S \cup S'|.\label{sarym2}\end{equation} 
We will denote the set of Sarymsakov matrices by $S_\text{SAR}$.
\end{definition}

This somewhat formal definition of a Sarymsakov matrix can be interpreted in terms of the graph associated with the matrix $P$, that is, the directed graph whose adjacency matrix is $P$. The consequent function $F_P(S)$ returns the set of out-neighbors of a set of nodes $S$. Condition \eqref{sarym1} corresponds to intersecting sets of neighbours and Condition \eqref{sarym2} corresponds to ``expansion" of sets of neighbours: the union of $F_P(S) \cup F_P(S')$ contains more elements than the union of the sets $S \cup S'$.

\begin{definition}[SIA Matrix] A matrix is called \emph{SIA}, i.e., \emph{Stochastic, Indecomposable and Aperiodic}, if it is stochastic and $$Q = \lim_{n \rightarrow \infty} P^n $$ exists and all the rows of $Q$ are the same. 
We will denote the set of SIA matrices by $S_\text{SIA}$.
\end{definition}
SIA matrices are called  \emph{regular matrices} in \cite{Seneta81}.

\subsection{Relation between the four Matrix Classes}
\label{sec:rel}
The four matrix classes that we have defined are included in one another:
\begin{equation} S_\text{PC} \subset S_\text{SCR} \subset S_\text{SAR} \subset S_\text{SIA} \label{inclusions}.\end{equation}
The first inclusion follows from the definitions and the second and third inclusions are proved  in  \cite[Section \textit{Bibliography and Discussion to \S\S4.3--4.4}]{Seneta81} and \cite[Section II]{Xia}.

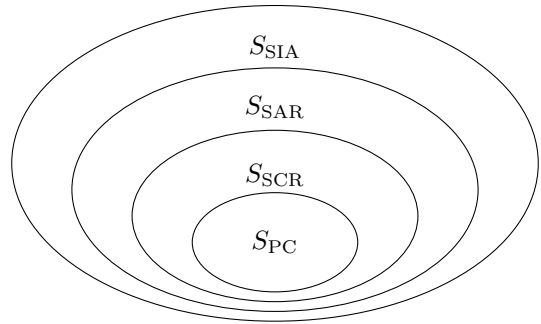
\begin{figure}[h!]
\centering
\def\firstcircle{(0,.15) ellipse (3.5cm and 2.1 cm)}
\def\secondcircle{(-90:.2cm)  ellipse (2.7cm and 1.62 cm)}
\def\thirdcircle{(-90:.55cm)  ellipse (1.9cm and 1.14 cm)}
\def\fourthcircle{(-90:.9cm)  ellipse (1.1cm and .66 cm)}

\begin{tikzpicture}
    \draw \firstcircle node[below] {};
    \draw \secondcircle node [above] {};
    \draw \thirdcircle node [below] {};
    \draw \fourthcircle node [below] {};

\node[align=center] at (0,-.9) {$S_\text{PC}$};
\node[align=center] at (0,0) {$S_\text{SCR}$};
\node[align=center] at (0,.9) {$S_\text{SAR}$};
\node[align=center] at (0,1.7) {$S_\text{SIA}$};

\end{tikzpicture}

\caption{Relation between different classes of stochastic matrices.}
\label{fourTypes}
\end{figure}

We present an alternative characterization of SIA matrices. 

\begin{prop}[Characterization of SIA Matrices] A stochastic matrix $P$ is SIA if and only if there is $p$ such that $P^p$ has a positive column.
\begin{proof} \textit{Only if.} Let $P$ be SIA. We prove that there is $p$ such that $P^p$ has a positive column. The $\lim_{t \to +\infty} P^t$ exists by definition of an SIA matrix. If there is no $p$ such that $P^p$ has a positive column, then  $\lim_{t \to +\infty} P^t$ has no positive column. On the other hand, 
all the rows of $\lim_{t \to +\infty} P^t$ are the same and the sum of the elements on each row is equal to 1, so that $\lim_{t \to +\infty} P^t$ has a positive column and we have a contradiction.

\textit{If.} Let $p$ be such that $P^p$ is positive-column. By inclusions \eqref{inclusions}, $P^p$ is SIA because positive-column matrices are SIA. Hence $$ \lim_{t \to +\infty} (P^p)^t$$ exists and has all rows equal, and therefore $$ \lim_{t \to +\infty} P^t$$ exists and has all rows equal and $P$ is SIA.
\end{proof}
\label{SIA_pos}
\end{prop}

This simple characterization has several consequences. It means that only the pattern of zero/nonzero elements determines if a matrix is SIA or not. 
Another consequence of Proposition \ref{SIA_pos} is the next proposition.

\begin{prop} For any set of stochastic matrices $\mathcal{S}$ the four statements are equivalent.
\begin{itemize}
    \item There exists an \textbf{SIA} product of matrices from $\mathcal{S}$
    \item there exists a \textbf{Sarymsakov} product of matrices from $\mathcal{S}$
    \item there exists a \textbf{scrambling} product of matrices from $\mathcal{S}$
    \item there exists a \textbf{positive-column} product of matrices from $\mathcal{S}$.
\end{itemize}
\begin{proof}
This is a consequence of Proposition \ref{SIA_pos} and of the inclusions \eqref{inclusions}.
\end{proof}
\label{prop:existence}
\end{prop}

We define the notion of an \emph{SIA set} as a set of stochastic matrices that has an SIA product (and therefore a Sarymsakov, a scrambling and a positive-column product). 

We now present a particular class of stochastic matrices that has been studied extensively in the context of automata theory \cite{volkov_survey}. We will call them automaton matrices in this article.
\begin{definition}[automaton matrix] A matrix is called \emph{automaton matrix} if it is  stochastic and has exactly one 1 on each row. 
\end{definition}

In the case of automaton matrices, the notions of Sarymsakov, scrambling and positive-column coincide, as shown in the next proposition.

\begin{prop} Let $P$ be an automaton matrix. The three properties are equivalent:
\begin{enumerate}[(i)]
    \item $P$ is Sarymsakov
    \item $P$ is scrambling
    \item $P$ is positive-column.
\end{enumerate} 
\begin{proof}
The inclusions in Equation \eqref{inclusions} imply (iii) $\Rightarrow$ (ii) and (ii) $\Rightarrow$ (i), so that (i) $\Rightarrow$ (iii) remains to be proved. Let $P$ be Sarymsakov and let $i,j \in \{1, \dots, n\}$. 
Condition \eqref{sarym2} cannot be satisfied for sets $S=\{i\}$ and $S'=\{j\}$ because $$|F_P(\{i\}) \cup F_P(\{j\})| > |\{i\} \cup \{j\}| = 2 $$  would imply $$|F_P(\{i\})| > 1 \text{ or } |F_P(\{j\})| > 1,$$ and at least one of the rows $i$ and $j$ would have more than one positive element, which is impossible by the definition of an automaton matrix. Therefore, Condition \eqref{sarym1} is satisfied for any pair of singletons $S=\{i\}, S'=\{j\}$, meaning that for any rows $i,j$, the positive element in row $i$ is in the same column as the positive element in row $j$, so that the matrix $P$ is in fact positive-column. 


\end{proof}
\label{coincide}
\end{prop}

\subsection{Indices and Bounds}
\label{sec:indices}
In the present section we will associate a natural combinatorial parameter with each of the previously discussed classes and provide bounds on them that depend only on the size of the matrices.
\begin{definition}[$\mathcal{X}$-index] Let $\mathcal{X}$ be a class of matrices. The \emph{$\mathcal{X}$-index} of a set of square matrices $\mathcal{S}$ is the smallest $\ell$ such that there is a product of length $\ell$ of matrices from $\mathcal{S}$ belonging to $\mathcal{X}$.
\end{definition}
The $\mathcal{X}$-index can be seen as the length of the shortest witness proving that a given set of matrices does actually belong to the class $\mathcal{X}$. If no product of $\mathcal{S}$  belongs to $\mathcal{X}$, we will agree that the $\mathcal{X}$-index is undefined.
In this section, we will focus on SIA-, Sarymsakov-, scrambling- and positive-column-indices of SIA sets and we will denote the corresponding index of a set $\mathcal{S}$ by $\sia(\mathcal{S})$, $\sar(\mathcal{S})$, $\scr(\mathcal{S})$, $\pc(\mathcal{S})$.
We note that our terminology is in agreement with the case of a single matrix, where the scrambling-index has received a lot of attention~\cite{AKELBEK20091099}.

Observe that for every SIA set $\mathcal{S}$, we have the following inequalities as a direct consequence of the inclusions \eqref{inclusions}: 
\begin{equation} 
\sia(\mathcal{S}) \leq \sar(\mathcal{S}) \leq \scr(\mathcal{S}) \leq \pc(\mathcal{S}).
\label{ineq} \end{equation}


One of the basic question arising in regard with these indices is how large the index of a set of $n \times n$ stochastic matrices can be? Such questions have received a lot of attention for the exponent~\cite[Section 3.5]{BrualdiRyser1991CombinatorialMatrixTheory} that can be seen as the $\mathcal{X}$-index, where $\mathcal{X}$ is the class of positive matrices. 
Likewise, the positive-column-index  has been studied extensively in the particular case of automaton matrices. It is called \emph{reset threshold} in this context.

\begin{thm} 
Let $\sia(n)$, $\sar(n)$, $\scr(n)$, $\pc(n)$ be the largest values of the corresponding indices among all SIA sets of $n \times n$ matrices.
For any dimension $n \in \mathbb{N}$, we have
 \begin{equation} \sia(n) \leq \sar(n) = \scr(n) = \pc(n). \label{ineq2}\end{equation}
\begin{proof} First, notice that we have the following inequalities a direct consequence of \eqref{ineq}:
\[ \sia(n) \leq \sar(n) \leq \scr(n) \leq \pc(n).\]
It remains to prove that $ \pc(n) \leq  \sar(n)$. Let $\pc_A(n)$ (resp. $\sar_A(n)$)  be the largest positive-column-index (resp. Sarymsakov-index) among all $n \times n$ sets of automaton matrices.
\begin{align}\pc(n) &= \pc_A(n) \label{e1}\\
&= \sar_A(n) \label{e2}\\
&\leq \sar(n) \label{e3}.
\end{align}
The first equality \eqref{e1} is Theorem~4 of~\cite{CDC}. The second \eqref{e2} is a consequence of the fact that an automaton matrix is Sarymsakov if and only if it is positive-column (see Proposition~\ref{coincide}). Inequality~\eqref{e3} holds true trivially, since automaton matrices form by definition a subset of stochastic matrices.
\end{proof}
\end{thm}

\subsection{From SIA to a Positive Column}
\label{sec:from_to}

In this section, we will provide upper bounds on the powers at which SIA matrices are guaranteed to have a positive column. First, we will prove in Proposition \ref{SIA_auto} that the $(n-1)$st power of any SIA \emph{automata} matrix has a positive column. Then, we will show that the $n^2 - 3n + 3$ power of any arbitrary SIA matrix has a positive column and present a finer bound that depends on the number of columns that are eventually positive (Theorem \ref{wielk}).


We will say that a matrix and a graph are \textit{associated} if the matrix is the adjacency matrix of the graph. 

\begin{prop} Let $P$ be an $n \times n$ SIA automaton matrix. Its $(n-1)$st power has a positive column. Moreover, the value $n-1$ cannot be decreased (in general): for any $n \in \mathbb{N}$ there exists an $n \times n$ SIA automaton matrix whose $n-2$ power has no positive column.
\begin{proof} By Proposition~\ref{SIA_pos}, there exists $p$ such that $P^p > 0$. The entry $i,j$ of matrix $P^p$ counts the number of walks of length $p$ from node $i$ to node $j$ in the graph $\mathcal{G}$ associated to $P$, and the positive column means that there is a walk of length $p$ from any node to the node corresponding to the positive column. 

\textit{Claim: $\mathcal{G}$ has no cycle other than a self-loop.} If there is a cycle of length $\geq 2$, the submatrix $C$ corresponding to the nodes of the cycle is a permutation matrix. Up to reordering of the nodes, the matrix $P$ is then equal to 
$$
\begin{pmatrix} A & B \\
0 & C
\end{pmatrix},
$$ with $C$ a permutation matrix. Then $P$ is not positive-column and we have a contradiction. Therefore there is no cycle of length $\geq 2$. Furthermore, $\mathcal{G}$ has a self-loop because there is no other cycle and all node have one outgoing edge.

It is now clear that the positive column of $P^p$ is the column corresponding to the node $i $ that has a self-loop. In $\mathcal{G}$, there is a walk of length at most $n-1$ from any node to $i$, and because node $i$ has a self-loop, there is a walk of length exactly $n-1$.

For the tightness part, the following matrix is an example of SIA automaton matrix whose $(n-2)$nd power has no positive column and its $n-1$ has one. 
$$\begin{pmatrix} 1 &  &  &  & \\ 
1 &  &&& \\
& 1 &&& \\
&& \ddots & &  \\
&&& 1 & 0\\
\end{pmatrix}.$$ 
\end{proof}
\label{SIA_auto}
\end{prop}

This last proposition has possible consequences for the \v{C}ern\'{y} conjecture because it means that $(n-1)\sia(n)$ is a bound on the reset threshold\footnote{Recall that the reset threshold is the positive-column index of a set of automaton matrices.} of synchronizing automata and therefore, any subquadratic bound would improve the state-of-the-art upper bound on the \v{C}ern\'{y} conjecture.

Proposition \ref{SIA_auto} gives a tight upper bound on the power at which an SIA automata matrix has a positive column. We will now generalize this result to all SIA matrices. 
In order to do so, we will rely on the \emph{local exponents} of primitive matrices. Recall that a square matrix $P$ is primitive if $P^t>0$ (entrywise) for some natural $t$. The smallest such $t$ is known as the \emph{exponent} $\exp(P)$ of $P$~\cite[Section 3.5]{BrualdiRyser1991CombinatorialMatrixTheory}.
The local exponents are refinements of this characteristic : the $k$th \emph{local exponent} $\exp_k(P)$ of a primitive matrix $P$ is the smallest power having at least $k$ positive rows~\cite{BruBol90}.
Observe that the first local exponent of $P$ is the positive-column index of the transpose of $P$ and the $n$th local exponent is $\exp(P)$. We will make use of the following theorem:

\begin{thm}[{\cite[Theorem 3.4]{BruBol90}}]
\label{thm:localExp}
The largest value of the $k$th local exponent among primitive $n \times n$ matrices, with $n \geq 2$, is equal to $n^2-3n+k+2$.
\end{thm}

Before stating our result, we define the notion of a column that is positive in sufficiently large powers of the matrix. 
\begin{definition}[Eventually Positive Columns] Let $P$ be a stochastic matrix. We say that a column $i$ is \emph{eventually positive} if there is a power $p$ such that the $i$th column of $P^p$ is positive. 
\end{definition}
We can already notice that the if $P^p$ has a positive $i$th column, then 
$P^{p+1} = P P^p $ has a positive $i$th column as well, and so $P^t$ has a positive $i$th column for any $t \geq p$.

\begin{thm} Let $P$ be an $n \times n$ SIA matrix, and let $k$ be the number of eventually positive columns of $P$. The power $k^2 - 4k + 3 + n$  of $P$ has a positive column.
\end{thm}
\begin{proof} We can assume that the first $k$ columns are eventually positive, and let us partition $P$ in blocks. 
$$ P=
\begin{pmatrix} A & B \\
C & D
\end{pmatrix},
$$ with $A$ having a size of $k \times k$ and blocks $B$, $C$ and $D$ having sizes $k \times (n-k)$, $(n-k) \times k$ and $(n-k) \times (n-k)$ respectively. 

\textit{Claim: $A$ is primitive and $B=0$.}
Let $p$ be such that $P^p$ has $k$ positive columns.  If we partition the matrix $P^p$ in blocks, as we did for $P$, we obtain
$$ P^p=
\begin{pmatrix}
E & F \\
G & H
\end{pmatrix},
$$ again with $E$ being $k \times k$. We have that $E$ and $G$ are entrywise positive because we assumed that the first $k$ columns of $P^p$ are positive.
$$P^{p+1} = 
\begin{pmatrix}
EA+FC & EB+FD \\
GA+HC & GB+HD
\end{pmatrix}.
$$
We prove now that $B=0$. Assume to the contrary that $b_{ij}$, the element at position $i,j$ in the matrix $B$, is positive. The $j$th column of $EB$ is equal to $$(EB)_j = \sum_i E_i b_{ij},$$ with $E_i$ being the $i$th column of $E$. Because $E$ is positive, the $j$th column of $EB$ is positive, and similarly, the $j$th column of $GB$ is positive. Therefore, the $(k+j)$th column of $P^p$ is positive, and we have a contradiction because only the first $k$ columns of $P$ are  eventually positive. Hence $B = 0$. Since $B = 0$, we have that $E = A^p$ and we can conclude that $A$ is primitive. 

Let us now compute $P^{k^2 - 4k + 3 + n} = P^{n-k} P^{k^2 - 3k + 3}$. Let us define $J$ and $K$ as below:
$$P^{n-k} = 
\begin{pmatrix}
A^{n-k} & 0 \\
J & D^{n-k}
\end{pmatrix}
$$ and 
$$P^{k^2 - 3k + 3} = 
\begin{pmatrix}
A^{k^2 - 3k + 3} & 0 \\
K & D^{k^2 - 3k + 3}
\end{pmatrix}
$$
The matrix $A^{k^2 - 3k + 3} $ has a positive column thanks to Theorem~\ref{thm:localExp} applied to the transpose of $A$. Let us assume that the $i$th column is positive. 

We can notice that both $J$ and $A^{n-k}$ have a positive element on each row. $A^{n-k}$ has a positive element on each row because $P$ has a positive element on each row and the first rows of $P$ are $\begin{pmatrix}A & 0\end{pmatrix}$. And
$J$ has a positive element on each row because there is a path of length $n-k$ from each node $k+1, \dots, n$ to some node $1, \dots, k$ in the graph associated to $P$. 

Therefore $J A^{k^2 - 3k + 3}$ and $A^{n-k} A^{k^2 - 3k + 3} $ have a positive $i$th column and therefore 
$$P^{k^2 - 4k + 3 + n}=
\begin{pmatrix}
A^{n-k} A^{k^2 - 3k + 3} & 0 \\
J A^{k^2 - 3k + 3} + D^{n-k}K  & D^{k^2 - 4k + 3 + n}
\end{pmatrix}
$$ has a positive column. 
\end{proof}
\label{wielk}


\begin{corollary}
The $(n^2 - 3n + 3)$rd power of a $n \times n$ SIA matrix has a positive column.
\end{corollary}
\begin{proof}
This is a consequence of the previous theorem and 
$$\forall n \in \mathbb{N}, \; k \in \{1, \dots, n\}, \; k^2 - 4k + 3 + n   \leq n^2 - 3n + 3.$$ 
\end{proof}

\section{The SIA-index}
\label{sec:sia-index}

In this section we study the SIA-index. 
We first show that, for any dimension $n$, the largest SIA-index, $\sia(n)$, is reached by a set of automaton matrices. We provide general upper and lower bounds on $\sia(n)$ that depend only on $n$. We describe the results of our numerical experiments to compute exactly $\sia(n)$ for small values of $n$. Finally, we conclude with algorithmic problems related to the computation of the SIA index.

The following result highlights the importance of automaton matrices in the context of SIA matrices.
\begin{prop} 
For any $n \in \mathbb{N}$, there is a set of matrices $\mathcal{B}_n$ such that $\sia(\mathcal{B}_n) = \sia(n)$ and every matrix of $\mathcal{B}_n$ is an automaton matrix. 
\begin{proof}
In order to construct $\mathcal{B}_n$ we will use techniques from \cite{BJO14}. Let $\mathcal{A}_n$ be a set of $n \times n$ stochastic matrices such that $\sia(\mathcal{A}_n) = \sia(n)$ (such a set exists by definition of $\sia(n)$).
Let \begin{equation}\mathcal{B}_n = \{B \in S_\text{AUTO} \; | \; \exists A \in \mathcal{A}_n, \; B \lesssim A\},\label{assoc_auto}\end{equation} 
 where $S_\text{AUTO}$ is the set of automaton matrices, and $B \lesssim A$ means that any zero element in $A$ must be zero in $B$: $a_{ij} = 0 \Rightarrow b_{ij} =0$. 
 
By definition of $\sia(n)$, we have $\sia(\mathcal{B}_n) \leq \sia(n)$. We prove $\sia(\mathcal{B}_n) \geq \sia(n)$.

\textit{Observation 1: The set $\mathcal{B}_n$ has an SIA product.} The set $\mathcal{A}_n$ has a positive-column product. We can use Lemma 1 of \cite{CDC}, stating that for any set of stochastic matrices $\mathcal{A}_n$ that has a positive-column product, the set $\mathcal{B}_n$, defined as in \eqref{assoc_auto} has a positive-column product as well. The positive-column product of the set $\mathcal{B}_n$ is thus SIA by inclusions \eqref{inclusions}. 

\textit{Observation 2: Any SIA product of $\mathcal{B}_n$ has a length larger than or equal to  $\sia(\mathcal{A}_n)$.} Indeed, if $B_1 \dots B_\ell$ is an SIA product, replacing each matrix $B_i$ by a matrix $A_i$ such that $B_i \lesssim A_i$ yields an SIA product of matrices from $\mathcal{A}_n$. 


Therefore, $\sia(\mathcal{B}_n) = \sia(n)$.
\end{proof}
\label{auto_stoch}
\end{prop}

\subsection{Bounds on the SIA-index}
\label{sec:bound}


\subsubsection{Upper bounds}


Propositions~\ref{SIA_auto} and~\ref{auto_stoch} tell us that there are tight connections between the SIA sets of matrices and synchronizing automata theory. Such connections were already utilized in
 \cite{CDC}, where the authors proved that the largest positive column index, $\pc(n)$, is equal to the largest
 reset threshold among automata of size $n$. 
 The best known upper bound on the reset threshold of automata of size $n$ is $\frac{n^3 - n}{6}$ \cite{Pin83a, Frankl82}. We can combine it with
 inequality \eqref{ineq2} to derive
 an upper bound on the SIA-index:
\[ \sia(n) \leq \frac{n^3 - n}{6}.\]
In fact, we believe that $\sia(n)$ is much smaller.
\begin{conj}
\label{conj:sia}
The SIA-index of a set of $n \times n$ stochastic matrices is bounded by $2n$.
\end{conj}

In the next subsections we will support this conjecture by providing results of computational experiments and analysis of sets that are extremal for the \v{C}ern\'{y} conjecture.
We believe that Conjecture~\ref{conj:sia} offers a new angle on the \v{C}ern\'{y} conjecture and can bring new insights. First, by Proposition~\ref{SIA_auto} the SIA-index of automaton matrices is the shortest $\ell$ such that there exists a positive-column product $P^{n-1}$, where $P$ is a product of length $\ell$. In other words, if Conjecture~\ref{conj:sia} is true, then a synchronizing automaton of size $n$ has a reset threshold at most $2n(n-1)$, which is a significant improvement of the state of the art.
Second, the SIA-index tends to be surprisingly small for automata with large reset thresholds highlighting the structural properties of these particular cases: Rystsov's automata \cite{RYSTSOV1997}, \v{C}ern\'{y} automata \cite{volkov_survey}, and other slowly synchronizing automata~\cite{AGV2013} have small SIA-indices.



\subsubsection{Numerical results}
We now present the results of our computational experiments that support our Conjecture~\ref{conj:sia}.
Since the bound on SIA-index for automaton matrices is equal to the bound on the SIA-index for stochastic matrices (Proposition \ref{auto_stoch}), we only investigate automaton matrices. We have computed on a computer cluster the SIA-index of all automata made of two matrices up to $n = 7$, and up to $n=9$ for \emph{initially connected automata}, a notion that we will define soon. The results are summarized in Table \ref{table1}.
We have done the same for all triplets of automaton matrices up to $n=5$ (Table \ref{table1}), and we obtain exactly the same maximum SIA-indices. 
The maximum SIA-index grows approximately like $2n$, as shown in Figure \ref{2n}.
Examples  of sets of pairs of  $8 \times 8$ and $9 \times 9$ matrices that have an SIA-index of $15$ and $16$ are depicted on Figure \ref{ex8}. 

\begin{table}[h!]
\centering
\begin{tabular}{l|p{3.3cm}|p{3.3cm}|p{3.3cm}}
   $n$ & maximum SIA-index, two matrices, all automata & maximum SIA-index, two matrices, IC automata. & maximum SIA-index, three matrices, IC automata\\
\hline
	1 & 0 & 0 & 0 \\
	2 & 1 & 1 & 1 \\
	3 & 3 & 3 & 3 \\
	4 & 5 & 5 & 5 \\
	5 & 8 & 8 & 8 \\
	6 & 10 & 10 & \\
	7 & 13 & 13 & \\
	8 &  & 15 & \\
	9 &  & 16 &
\end{tabular}
\caption{Exhaustive tests for pairs and triplets of automaton matrices. The first column is the size of the matrices, the second column is the maximum SIA-index of two-matrices automata, the third column is the maximum SIA-index of two-matrices initially connected automata (see Definition \ref{IC})}, and the fourth column is the maximum SIA-index of three-matrices automata.
\label{table1}
\end{table}

\begin{figure}[h!]
\centering
\includegraphics[scale=.7]{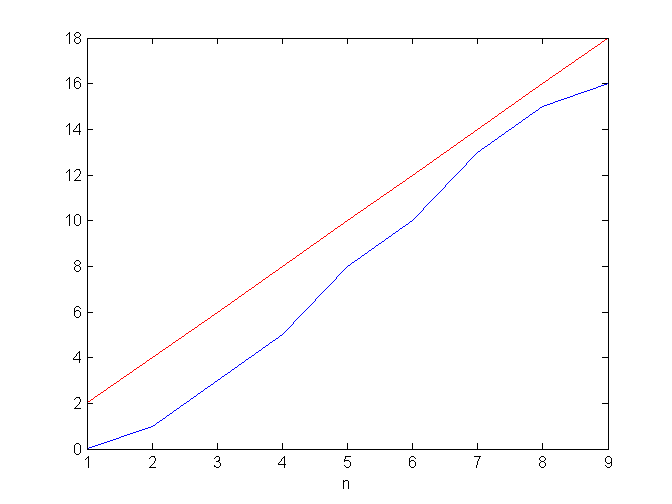}
\caption{Maximum SIA-index for pairs of matrices (blue) and the curve $y=2n$ (red).}
\label{2n}
\end{figure}

\begin{figure}[h!]
\centering
\begin{subfigure}{7 cm}
  \centering
  \includegraphics[scale=.5]{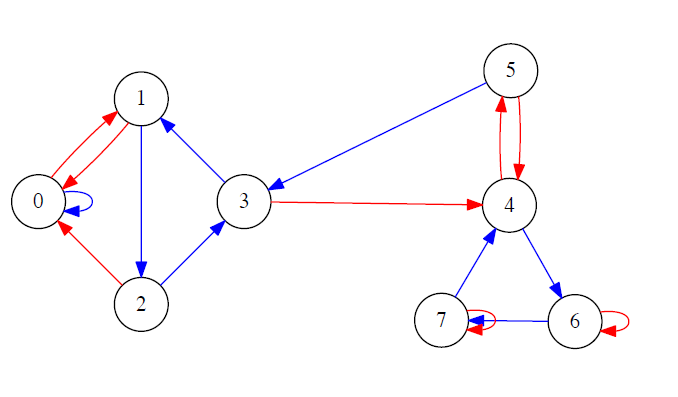}
  \caption{Pair of $8 \times 8$ matrices that has an SIA-index of 15. This is the only such pair up to relabelling of the nodes.}
  \label{fig:sub1}
\end{subfigure}%
\hspace{.5cm}
\begin{subfigure}{7cm}
  \centering
  \includegraphics[scale=.5]{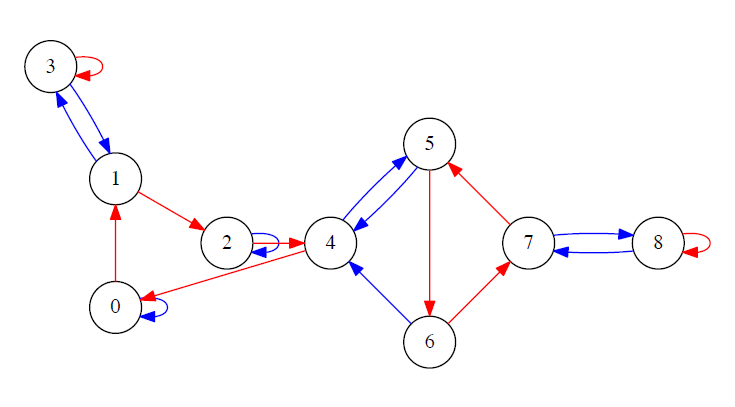}
  \caption{Pair of $9 \times 9$ matrices that has an SIA-index of 16. There are 12 different (up to relabelling of the nodes) sets that have an SIA-index of 16.}
  \label{fig:sub2}
\end{subfigure}
\caption{Examples of matrix sets that have extremal SIA-index. The matrices are represented by their graphs, the blue edges represent the positive elements of one matrix and the red edges the other matrix. }
\label{ex8}
\end{figure}

Our methodology was the following. For the exhaustive tests, we have enumerated all sets of automaton matrices and for each 
set we have computed its SIA-index.
In order to compute the SIA-index, we have enumerated
all matrix products corresponding to Lyndon words (defined below) and looked for the shortest SIA product of this form. 
Proposition \ref{prop:Lyndon} guarantees that the correct SIA-index is computed.


Let us now define Lyndon words. We will treat now the elements of a matrix set $\mathcal{S}=\{A_1, A_2, \ldots, A_k\}$ as \emph{symbols} (or \emph{letters}) that are linearly ordered in an arbitrary fashion, e.g. $A_1 \prec A_2 \prec \ldots \prec A_k$. Clearly, products of matrices from $\mathcal{S}$ are in one-to-one correspondence with \emph{words} -- the sequences of symbols $A_1,A_2,\ldots,A_k$. 
The \emph{lexicographic order} on the words is defined as follows: $P \prec Q$ if and only if either  $Q = PU$ for some word $U$; or $P=UA_iV$ and $Q=UA_jW$ for $A_i \prec A_j$ and some words $U,V,W$.

\begin{definition}[Lyndon word] A \emph{cyclic shift} of a word $P$ is a word of the form $VU$ if $P=UV$.
A non-empty word $P$ is Lyndon if it is strictly smaller in the lexicographic order than all of its cyclic shifts.
\end{definition}


\begin{prop}
\label{prop:Lyndon}
Let $\mathcal{S}$ be an SIA set of matrices such that $\sia(\mathcal{S}) = \ell$. There is a Lyndon word of length $\ell$ such that the corresponding product of matrices from $\mathcal{S}$ is SIA.
\end{prop}
\begin{proof}

Let $P$ be a word corresponding to an SIA product of length $\ell$.
Observe first that any cyclic shift of $P$ leads to an SIA product as well: if the powers of the matrix associated with $P=UV$ converge to the rank one matrix $Q$, then the matrices associated with $V(UV)^tU$ converge to rank one matrix $VQU$ as $t \to \infty$.   
Thus, we can always choose a word $P'$ corresponding to an SIA product that is smaller or equal than any of its cyclic shifts in the lexicographic order. 
Secondly, if $P'$ has the form $P'=U^t$ for some $t$, then $U$ clearly corresponds to an SIA product. Since $\sia(\mathcal{S}) = \ell$, we immediately conclude that $t=1$, i.e. $P'$ is aperiodic.
 Now we can invoke the classical result about the Lyndon words stating that an aperiodic word that is not larger than any of its cyclic shifts is actually Lyndon~\cite[Proposition 5.1.2]{lothaire1997combinatorics} and the proposition follows. 
\end{proof}

We have noticed in our tests that all extremal examples are \emph{initially connected} (in fact, even strongly connected). Therefore, we have decided to analyze larger values of $n$ 
restricted to the case of 
initially connected automaton matrices. This has allowed us to perform the tests up to  $n = 9$ instead of $n = 7$.
\begin{definition}[Initially Connected] A set of automaton matrices $\{P_1, \dots, P_m\}$ is called \emph{initially connected} or \emph{IC} if in the graph associated with the matrix $P_1 + \dots + P_m$  there exists a node $q$ such that there is a path from $q$ to any node of the graph. In particular, if the graph associated to $P_1 + \dots + P_m$ is strongly connected, the set is initially connected. 
\label{IC}
\end{definition}

\subsubsection{Lower bounds}
Proposition~\ref{SIA_auto} implies that automaton matrices with a small SIA-index have a small reset threshold as well. Therefore, we focused on automata that are known to be tight for the \v{C}ern\'{y} conjecture. The results are summarized in~Fig.\ref{fig:CernyTightSIA}.
A list of automata that are known to be tight for the \v{C}ern\'{y} conjecture can be found in \cite{volkov_survey}.

\begin{figure}[h!]
\centering
\vspace{.5cm}
\begin{tabular}{l|l}
   Automaton & SIA-index \\
\hline
   \v{C}ern\'{y} family (with $n\geq 3$) & $n$  \\
   3 states (3 different automata) & 3 \\
   4 states (3 automata) & 5 \\
   5 states & 7 \\
   6 states (Kari automaton) & 9
\end{tabular}
\vspace{.5cm}
\caption{SIA-indices of automata that are known to be tight for the \v{C}ern\'{y} conjecture.}
\label{fig:CernyTightSIA}
\end{figure}

The \v{C}ern\'{y} family is the only known infinite series of automaton matrices with the shortest positive-column product of length $(n-1)^2$~ \cite{volkov_survey}. For $n=4$, the \v{C}ern\'{y} set is equal to 
$$\mathcal{C}_n = \left \{A=
\begin{pmatrix}
0 & 1 & 0 & 0 \\
0 & 0 & 1 & 0 \\
0 & 0 & 0 & 1 \\
1 & 0 & 0 & 0
\end{pmatrix}, B = 
\begin{pmatrix}
0 & 1 & 0 & 0 \\
0 & 1 & 0 & 0 \\
0 & 0 & 1 & 0 \\
0 & 0 & 0 & 1
\end{pmatrix}
\right\}.
$$ We will see in the following proposition that its shortest SIA product is equal to $AB^{n-1}$, which implies  $\sia(n) \geq n$. For small values of $n$, the values computed in Table \ref{table1} provide slightly better lower bounds.

\begin{prop}
The \v{C}ern\'{y} set $\mathcal{C}_n$ of matrices of dimension $n$ has an SIA-index of $\sia(\mathcal{C}_n) = n$.
\end{prop}
\begin{proof} 
Observe first that $(AB^{n-1})^{n-2}A$ has a positive-column and $AB^{n-1}$ is SIA. Thus, $\sia(\mathcal{C}_n) \leq n$.
Furthermore, since the shortest positive-column product of  $\mathcal{C}_n$ has length $(n-1)^2$, see e.g.~\cite{Gu13}, and due to Proposition~\ref{SIA_auto} we have $(n-1) \sia(\mathcal{C}_n)  \geq \pc(\mathcal{C}_n)$. Therefore, $\sia(\mathcal{C}_n) \geq (n-1)$.

It remains to show now that the case $\sia(\mathcal{C}_n) = (n-1)$ is impossible. Assume to the contrary that $P$ is an SIA product of length $n-1$. Therefore, $P^{n-1}$ has a positive column (by Proposition~\ref{SIA_auto}). By~\cite[Proposition 4]{Gu13} every positive-column product of $\mathcal{C}_n$ has at least of $n^2-3n+2$ occurrences of $A$ and $n-1$ occurrences of $B$, thus, $P$ has exactly one occurrence of $A$ and $n-2$ occurrences of $B$. Applying Proposition~\ref{prop:Lyndon} we further conclude that the word $A^{n-2}B$ corresponds to an SIA product as well, which is not the case.



\end{proof}

Now we will analyze a set of matrices derived from the Wielandt series of matrices that have the largest possible  exponent among $n \times n$ matrices~\cite[Chapter 3.5]{BrualdiRyser1991CombinatorialMatrixTheory}.
Matrix sets of this kind often appear in the study of combinatorial characteristics of matrix sets, e.g. generalizations of the exponents~\cite{SHADER03} or in the study of positive-column indicies~\cite{AGV2013}. We define the Wielandt set of automaton matrices as
$$\mathcal{W}_n = \left\{A=\begin{pmatrix}
& 1 &&& \\
&& 1 && \\
&&& \ddots & \\
&&&& 1 \\
0 & 1 &&
\end{pmatrix}, B=\begin{pmatrix}
& 1 &&& \\
&& 1 && \\
&&& \ddots & \\
&&&& 1 \\
1 & 0 &&
\end{pmatrix}\right\}, $$ where the omitted elements are zeros.

\begin{prop} The Wielandt set of matrices $\mathcal{W}_n$ has an SIA-index of $\sia(\mathcal{W}_n) = n-1$.
\end{prop}
\begin{proof}
Recall that the shortest positive-column product of  $\mathcal{W}_n$ has length $n-3n+3$~\cite[Theorem 2]{AGV2013}.
Since $(n-1) \sia(\mathcal{W}_n)$ cannot be strictly smaller than $n-3n+3$ by Proposition~\ref{SIA_auto}, we immediately conclude that $\sia(\mathcal{W}_n) \geq (n-1)$. This bound is tight, since $AB^{n-2}$ is the desired SIA product.
\end{proof}


\subsection{Complexity Results}
\label{sec:complexity}

In the present section we address algorithmic problems related to SIA sets and the SIA-index. We assume that the reader is familiar with the computational complexity theory.

\subsubsection{Checking the Existence of an SIA Product}

Proposition \ref{prop:existence} states that for a set $\mathcal{S}$, the existences of a positive-column product, of a scrambling product, of a Sarymsakov product, and  of an SIA-product are equivalent.
Therefore, deciding algorithmically the existence of an SIA-product  can be done in polynomial time with techniques developed in \cite[Section 5]{PV} and in \cite{CDC} to decide the existence of a positive-column product.

\subsubsection{Computing the SIA-index}

Now we are going to show that the problem of finding the SIA-index of a given set of matrices is computationally hard. Furthermore, it is hard even for matrices with positive diagonals  -- a special case frequently arising in consensus applications.
More precisely, we will present 
a reduction
from the Boolean satisfiability problem to the problem of computing the SIA-index of sets of matrices 
with the following zero pattern:
$$\begin{pmatrix} + &&& \\
\cdot & + && \\
\vdots && \ddots & \\
\cdot &&& +  
\end{pmatrix},$$ where $+$ denotes a positive element and $\cdot$ denotes an element that can be either positive or zero. The following lemma establishes that such a matrix is SIA if and only if it is positive-column.

\begin{lem}
Let $P$ be a stochastic matrix that has a positive diagonal and all elements that belong neither to the diagonal nor to the first column are 0. $P$ is SIA if and only if it is positive-column.
\begin{proof}
The "if" part is a consequence of inclusions (\ref{inclusions}). 
We prove the "only if" part. If $P$ is SIA, there is $k$ such that $P^k$ is positive-column. If $k = 1$, $P$ is positive-column and the proof is complete. So we are left with the case $k\geq 2$.
It is clear that only the first column of $P^k$ can be positive and that $A_{11} > 0$, so we prove that for any $i \in \{2, \dots, n\}$, $P_{i1} >0$.

For any $i \in \{2, \dots, n\}$, the element $(P^k)_{i1}$ is positive only if $P_{i1}$ or $(P^{k-1})_{i1}$ is positive. If $A_{i1}$ is positive, the proof is complete. If not, we use $(P^{k-1})_{i1} > 0$ and we repeat the same argument, so that either $P_{i1}>0$ or $(P^{k-2})_{i1}>0$ and we can continue iteratively until we arrive at the conclusion that either $P_{i1}>0$ or $P_{i1}>0$.
\end{proof}
\label{lem:diag}
\end{lem}

\begin{thm} 
For a given integer $k$ and a set $\mathcal{S}$ of stochastic matrices with positive diagonal, the problem of deciding whether $\sia(\mathcal{S}) \leq k$ is NP-complete.
\begin{proof} We use a reduction from 3-SAT similar to that of Theorem 8 in \cite{eppstein90reset}. 

A 3-SAT instance consists of variables $X_1, \dots, X_v$, clauses $C_1, \dots, C_c$ of the form $L_1 \lor L_2 \lor L_3$ where each $L_i$ is a literal, that is, either a variable or the negation of a variable. The problem is to determine whether the formula $C_1 \land \dots \land C_c$ is satisfiable. 

Given a 3-SAT formula $F$ with $v$ variables, we construct a set of matrices that has an SIA product of length smaller or equal to $v$ if and only if the 3-SAT formula is satisfiable. 

We define a set $\mathcal{S}$ of $2v$ matrices. 
The matrices have size $(1+c+v) \times (1+c+v)$. 
The matrix $A_{X_i}$, representing literal $X_i$, has ones on the diagonal, a 1 in position $(1+i, 1)$ and a 1 in position $(1+v+j,1)$ for every $j$ for which assigning $X_i$ to TRUE satisfies clause $C_j$. 
The matrix $A_{\lnot X_i}$  has ones on the diagonal,  a 1 in position $(1+i, 1)$ and a 1 in position $(1+v+j,1)$ if assigning $X_i$ to FALSE satisfies clause $C_j$. 

Theorem 6 of \cite{CDC} proves that the formula $F$ is satisfiable if and only if $\mathcal{S}$ has a positive-column product of length at most $v$. By Lemma \ref{lem:diag}, $\mathcal{S}$ has a positive-column product of length at most $v$ if and only if it has an SIA-product of length at most $v$, which concludes the proof of NP-hardness.

Additionally, it is clear that the problem belongs to NP. 
\end{proof}
\end{thm}

By the same argument we can conclude that the problems of deciding whether there exists a \emph{Sarymsakov}, \emph{scrambling}, or \emph{positive-column} product of length at most $\ell$ are NP-complete as well.
The previous NP-completeness result holds for matrices with positive diagonals. If we remove the restriction that the matrices have positive diagonals, we can obtain a stronger infeasibility result.  
\begin{thm}
For every $\alpha>0$, it is NP-hard to approximate the SIA-index of sets of automaton matrices of size $n \times n$ within a factor $(1 - \alpha) \log (n)$. 
\end{thm}
\begin{proof}

Our proof is based on 
the reduction given in~\cite[Section 3]{GH2011} showing inapproximability of the reset thresholds, i.e., positive-column indices of automaton matrices. Recall that the classical \emph{set cover problem} is formalized as follows.
Given a set of elements $U=\{1, 2, \ldots, n\}$ and a collection $\mathcal{F}$ of subsets of $U$ whose union is equal to $U$, i.e., $U = \cup_{T \in \mathcal{F}} T$. The set cover problem is to compute a sub-collection $\mathcal{F}' \subseteq \mathcal{F}$ of the smallest possible size, whose union still equals to $U$. This problem is computationally hard, namely, for every $\alpha>0$, there is no $(1 - \alpha) \log (n)$ approximation algorithm computing the set cover, unless $P=NP$~\cite{Dinur14}.

The reduction from the set cover problem to our problem 
is done as follows. For every set $T \in \mathcal{F}$ we construct a $(0,1)$-matrix $A_T$ of size $(n+1) \times (n+1)$:
$A_T(i,j)=1$ if and only if either $i \in T$ and $j = n + 1$, or $i \notin T$ and $j = i$. Clearly, $A_T$ is stochastic and automatic. Let $\mathcal{S} = \{A_T \mid T \in \mathcal{F}\}$. 
It is easy to see that a product $P=A_{T_1}A_{T_2}\ldots A_{T_\ell}$ of matrices from $\mathcal{S}$ is SIA if and only if the $(n+1)$st column of $P$ is positive. Observe now that the latter condition is equivalent to \[U=\bigcup_{i=1,\ldots,\ell} T_i.\] 
Therefore, 
$\sia(\mathcal{S})$ is equal to the size of the smallest cover of $U$ and the result follows.

\end{proof}

\section*{Conclusion}
In the present paper we studied SIA sets of stochastic matrices and compared them to other classical matrix classes.
We introduced the SIA-index for such sets and related it to Lyndon words and well-studied combinatorial characteristics such as local exponents and the scrambling index. Bounds on the SIA-index in terms of matrix sizes were provided and it was shown that $\sia(\mathcal{S})$ is hard to compute for a given matrix set $\mathcal{S}$. 

Several problems remain open. First, there is a significant gap  between the lower and upper bounds on $\sia(n)$. As we have seen, this problem is related to the long-standing \v{C}ern\'{y} conjecture from automata theory and might be difficult to resolve completely. Another open problem is to find the best approximation ratio achievable by polynomial-time algorithms computing the SIA-index of a given matrix set. Since our contribution is based on synchronizing automata theory, the methods used in~\cite{Gawrychowski2015} can potentially be used to establish the exact ratio. 
We have proved that for each matrix size $n$ there is an SIA set of automaton matrices having the largest possible SIA-index among stochastic matrices of size $n \times n$.
We wonder whether representing the largest SIA-index as the solution of an optimization problem on the space of stochastic matrix sets can lead to another proof of this result. Such formulation can potentially unify and generalize similar results appearing in~\cite{GGJ16,BJO14,CDC}.

\bibliography{references}{}
\bibliographystyle{plain}

\end{document}